\documentclass[11pt]{amsart}

\usepackage{amsmath}
\usepackage{amsfonts}
\usepackage{amssymb}
\usepackage{amsthm,epic}
\usepackage{mathrsfs}
\usepackage[utf8]{inputenc}
\usepackage[T1]{fontenc} 
\usepackage{xcolor}
\usepackage{url}
\usepackage[plainpages=false,pdfpagelabels,colorlinks=true,citecolor=blue,hypertexnames=false]{hyperref}
\usepackage{fullpage}

\newtheorem{theorem}{Theorem}[section]
\newtheorem{lemma}[theorem]{Lemma}

\newtheorem{cor}[theorem]{Corollary}
\newtheorem{prop}[theorem]{Proposition}

\theoremstyle{definition}
\newtheorem{example}{Example}

\usepackage{paralist} 
\def\paragraph#1{\smallskip\noindent{\bf #1.}}
\newcommand{\mV}{\ensuremath{\mathcal{V}}}
\newcommand{\bz}{\ensuremath{\mathbf{z}}}

\newcommand{\bk}{\ensuremath{\mathbf{k}}}

\newcommand{\mD}{\ensuremath{\mathcal{D}}}
\newcommand{\bzeta}{\ensuremath{\boldsymbol \zeta}}

\def\Res{\operatorname{Res}}
\def\QQ{\mathbb{Q}}

\newcommand{\bzer}{\ensuremath{\mathbf{0}}}

\newcommand{\bff}{\ensuremath{\mathbf{f}}}

\author{Stephen Melczer}
\author{Bruno Salvy}

\title{Symbolic-Numeric Tools for Analytic Combinatorics in Several Variables}
\address[S. Melczer]{Cheriton School of Computer Science, University of Waterloo, Waterloo, ON, Canada \& LIP (U. Lyon, CNRS, ENS Lyon, UCBL), Lyon, France}
\email{smelczer@uwaterloo.ca}
\address[B. Salvy]{Inria, LIP (U. Lyon, CNRS, ENS Lyon, UCBL), Lyon, France}
\email{Bruno.Salvy@inria.fr}


\begin{document}

\begin{abstract}
Analytic combinatorics studies the asymptotic behavior of sequences through the analytic properties of their generating functions.
This article provides effective algorithms required for the study of analytic combinatorics in several variables, together with their complexity analyses. Given a multivariate rational function we show how to compute its smooth isolated critical points, with respect to a polynomial map encoding asymptotic behaviour, in complexity singly exponential in the degree of its denominator. We introduce a numerical Kronecker representation for solutions of polynomial systems with rational coefficients and show that it can be used to decide several properties (0 coordinate, equal coordinates, sign conditions for real solutions, and vanishing of a polynomial) in good bit complexity. Among the critical points, those that are minimal---a property governed by inequalities on the moduli of the coordinates---typically determine the dominant asymptotics of the diagonal coefficient sequence.  When the Taylor expansion at the origin has all non-negative coefficients (known as the `combinatorial case') and under regularity conditions, we utilize this Kronecker representation to determine probabilistically the minimal critical points in complexity singly exponential in the degree of the denominator, with good control over the exponent in the bit complexity estimate.  Generically in the combinatorial case, this allows one to automatically and rigorously determine asymptotics for the diagonal coefficient sequence. Examples obtained with a preliminary implementation show the wide applicability of this approach.
\end{abstract}

\maketitle

\section{Introduction}
\label{sec:Intro}
The starting point of analytic combinatorics is an integral representation for the Taylor coefficients of a function analytic at the origin: if 
\[ F(z_1,\dots,z_n)=\sum_{k_1,\dots,k_n}c_{k_1,\dots,k_n}z_1^{k_1}\dotsm z_n^{k_n},\]
then 
\[ c_{k_1,\dots,k_n}=\frac{1}{(2\pi i)^n}\int_T{F(z_1,\dots,z_n)\frac{dz_1\dotsm dz_n}{z_1^{k_1+1}\dotsm z_n^{k_n+1}}},\]
where $T$ is a torus around the origin defined by $z_i=r_ie^{2\pi i \theta_i}$ for sufficiently small $r_1,\dots,r_n$. Of special interest is the sequence of \emph{diagonal} coefficients
\begin{equation}\label{eq:integral}
A_k:=c_{k,\dots,k}=\frac{1}{(2\pi i)^n}\int_T{F(z_1,\dots,z_n)\frac{dz_1\dotsm dz_n}{(z_1\dotsm z_n)^{k+1}}}.
\end{equation}
The principle of analytic combinatorics is to study this integral by deforming the contour $T$ in such a way that it avoids the singularities of the integrand and tends to \emph{critical points} where the integral concentrates as $k\rightarrow\infty$. A variety of techniques for analytic functions in one variable ($n=1$) are presented in the book by Flajolet and Sedgewick~\cite{FlSe09}. More recently, Pemantle and Wilson have given analogous tools in the more involved multivariate situation~\cite{PeWi13}.

Our aim in this work is to launch the study of analytic combinatorics in several variables from the point of view of computer algebra and complexity, in the case when $F(\bz)$ is a rational function\footnote{We use boldface letters for $n$-tuples: $\bz=(z_1,\dots,z_n)$, $\bk=(k_1,\dots,k_n)$, $\bz^\bk=z_1^{k_1}\dotsm z_n^{k_n}$, $\bz^k=(z_1\dotsm z_n)^{k}$, \dots} $G(\bz)/H(\bz)\in\mathbb{Q}(\bz)$.

\begin{example}\label{ex:apery1}
The sequence $A_k=\sum_{i=0}^k{\binom{k}{i}^2\binom{k+i}{i}^2}$ of Ap{\'e}ry numbers  is, like all multiple binomial sums, the diagonal of a rational function~\cite{BostanLairezSalvy2015}; in this case, their generating function is the diagonal of $1/H(a,b,c,z)$ where
\[H(a,b,c,z) = 1-z(1+a)(1+b)(1+c)(abc+bc+b+c+1).\]
With this rational input, our algorithm produces automatically the asymptotic behaviour
\begin{equation}\label{eq:asymptApery}
A_k = \frac{(17+12\sqrt{2})^k}{k^{3/2}} \cdot \frac{\sqrt{34+24\sqrt{2}}}{8\pi^{3/2}}(1+O(1/k)).
\end{equation}
See Example~\ref{ex:apery2} below for details.  In general, it is not possible to provide such an explicit closed form for the quantities involved in the asymptotic behaviour. The output will then be a combination of an exact symbolic representation and a precise numerical estimate:
\[A_k={(33.97056274\ldots)^k}{k^{-3/2}}(.2200437670\ldots+O(1/k)).\]
\end{example}

\begin{example}\label{ex:pattern-avoiding}
The $k^\text{th}$ coefficient $A_k$ of the diagonal of the rational function
\[ F(x,y) := \frac{1-x^3y^6+x^3y^4+x^2y^4+x^2y^3}{1-x-y+x^2y^3-x^3y^3-x^4y^4-x^3y^6+x^4y^6}, \]
is the number of binary words with $k$ zeroes and $k$ ones that do not contain 10101101 or 1110101 \cite[\S5.3]{PeWi08}. Here the numerical output of our algorithm has the same form,
\[A_k=(.6029459856\ldots){(3.910193204\ldots)^k}{{k^{-1/2}}}(1+O(1/k))\]
with its symbolic form 
\[A_k=\frac{u^{-k}}{\sqrt{k\pi}}A(u)\sqrt{B(u)}(1+O(1/k)),\quad P(u)=0,\]
for explicit polynomials $P,A,B$ in $\mathbb{Z}[u]$ of degree at most~20 and $u\approx 0.2557418388$ approximating a specific root of $P$.
\end{example}

The method described by Pemantle and Wilson can be summarized as follows:
\begin{itemize} \itemsep=2mm
 \item[(i)] \emph{Locate the critical points} of the map $\phi:(\bz)\mapsto z_1\dotsm z_n$ on the algebraic set $\mV(H):=\{\bz\mid H(\bz)=0\}$. Assuming $\mV(H)$ is a smooth manifold, the critical points are defined by the family of equations given by $H=0$ and the vanishing of the set of maximal minors of the Jacobian matrix of $[H,\phi]$. Since we also need to avoid the algebraic set $\mathcal{V}(\phi):=\{\bz\mid\phi(\bz)=0\}$, this system is equivalent to the vanishing of
\begin{equation}\label{eq:sys-critical}
\left(H,z_1\frac{\partial H}{\partial z_1}-\lambda,\dots,z_n\frac{\partial H}{\partial z_n}-\lambda\right)
\end{equation}
after introducing a Lagrange multiplier $\lambda$.
 \item[(ii)] \emph{Filter out} those critical points that do not contribute to the asymptotic behavior of \eqref{eq:integral}. In particular, keep only those critical points that can be reached by deformation of the contour $T$ outside of $\mathcal{V}(H)\cup\mathcal{V}(\phi)$.
 \item[(iii)] \emph{Compute the local behavior} of $F$ at the contributing critical points.
 \item[(iv)] \emph{Transfer} into the asymptotic behavior of the integral.
\end{itemize}
\vspace{0.1in}

The theory, as it has developed, has shown on one hand how difficult analyzing the integral~\eqref{eq:integral} in full generality can be while at the same time detailing large classes of problems to which systematic techniques apply. Although the methods developed by Pemantle and Wilson are quite explicit, their arguments often rely on some human intervention, typically in the filtering step, which can be quite involved.

\paragraph{Our Contributions}
This work gives, to our knowledge, the first completely automatic treatment and complexity analysis for the asymptotic enumeration of a large class of multivariate problems.

From a computer algebra perspective the critical points are simply the solutions of the polynomial system~\eqref{eq:sys-critical}, but the process of filtering the critical points involves deciding inequalities of moduli, which relies on techniques that are typically more expensive. Our approach is to combine robust numerical evaluation techniques and exact polynomial representations with a tight control over the degree and the bit size of the coefficients. 

We make generic regularity assumptions which result in zero dimensional polynomial systems. The solutions of the system are given by a Kronecker representation (also known as a rational univariate representation) which consists of a square-free univariate polynomial $P(u)$ and a param\-etrization of the coordinates of the form $Q(u)/P'(u)$. The size of the integer coefficients appearing in these polynomials obey known bounds. Using this tool we define a \emph{numerical Kronecker representation}, where the roots of $P$ are evaluated numerically with sufficient precision to determine which (if any) of the coordinates are 0 or equal between two or more solutions, and to determine their sign in the case of real solutions. We show that, under regularity assumptions, such a representation can be computed in good complexity (Theorem~\ref{th:fsolvekro}). More precisely, if a system of $n$ polynomials of degree at most $d$ and integer coefficients of height (bit size) at most $h$ in $n$ variables forms a reduced regular sequence, then the numerical Kronecker representation can be computed in~$\tilde{O}(hd^{3n})$ bit operations by a probabilistic algorithm. Since generically the number of solutions is $d^n$, this can be seen as polynomial complexity in a natural size of the problem. More precise estimates taking into account the complexity of evaluation of the polynomial system (hence for instance its sparseness) can also be given. The idea of using such a representation in order to reduce numerical computations to the resolution of univariate polynomials is not new. However, to the best of our knowledge, the connection between the good properties of the Kronecker representation in terms of bit size and the fast and precise algorithms operating on univariate polynomials has not been made before, except in the case of bivariate systems~\cite{BouzidiLazardPougetRouillier2015,KobelSagraloff2015}.

\paragraph{Previous work}
DeVries et al.~\cite{DevHoPe11} gave an algorithm going from $F(x,y)$ to diagonal asymptotics in the bivariate case assuming that the critical points are isolated. Raichev~\cite{Raic12} developed a Sage package which computes the asymptotic contribution of a non-degenerate critical point on the border of the domain of convergence at which $\mV(H)$ is smooth or the transverse intersection of smooth algebraic varieties, however the users of this package must know \emph{a priori} that the critical points they examine are the ones which contribute to dominant asymptotics (under stronger hypotheses than we require this can be used to give rigorous asymptotics). Neither of these works provides a complexity analysis.

Another approach to the asymptotic behaviour of the integral~\eqref{eq:integral} is to compute a differential equation for the generating function of the diagonal by creative telescoping. This equation is Fuchsian, so that the possible behaviour of the generating function at its dominant singularity is known and a list of possible asymptotic behaviours for the coefficients can be deduced. However, the constant factor in front of these asymptotic behaviours is not provided by this method, which does not even let one decide whether this coefficient is zero. This approach has good complexity~\cite{BostanLairezSalvy2013} and can be combined with the present work in order to provide full asymptotic expansions of the coefficients efficiently.

Concerning polynomial systems, the literature is vast~\cite{Mora2003}. Our work is in the filiation of works by Giusti, Heintz, Pardo and collaborators~\cite{GiustiHeintzMoraisMorgensternPardo1998,GiustiLecerfSalvy2001,KrickPardoSombra2001,Schost2001} on the use of the Kronecker representation in complex or real geometry, going far beyond the simple systems we consider here. Relevant works using the alternate denomination rational univariate representation~\cite{Rouillier1999,BPR06} also exploit numerical evaluation and tight bounds on the distance between solutions~\cite{EmirisMourrainTsigaridas2010}. These bounds are comparable to ours when expressed in terms of $d^n$, but they exploit more information on the input, such as mixed volume, that we do not take into account at this stage.

\paragraph{Notation for complexity estimates}
We start with $H \in \mathbb{Z} [\bz]$  of degree $d$ where $\bz=(z_1,\dots,z_n)$. In all complexity estimates, $d \geqslant 2$  is assumed and $D := d^{n}$.

For $f$ and $g$ two functions defined and positive over~$(\mathbb{N}^\star)^m$,  the notation $f(a_1,\dots,a_m) = O(g(a_1,\dots,a_m))$ is used in complexity estimates when there exists a constant~$K$ such that $f(a_1,\dots,a_m)\le Kg(a_1,\dots,a_m)$ in this domain.

The sizes and complexity of our operations grow like $D^c$, where $c$ is an exponent that we want to bound tightly. We make use of the usual $\tilde{O}$ notation to avoid technical discussions on the subexponential terms. By definition, $f=\tilde{O}(g)$ when $f=O(g\log^kg)$ for some $k\ge0$. For instance, $O(nD)=\tilde{O}(D)$ since we assume $d \geqslant 2$. 

\paragraph{Plan}
This article is structured as follows. We first give a brief summary of the work of Pemantle and Wilson that we use. Next, we present our symbolic-numeric tools in Section~\ref{sec:symnum}. This is then applied in Section~\ref{sec:CPs} to the computation of the contributing critical points in the combinatorial case. We conclude with a few experiments using a preliminary implementation in Maple.

\section{Analytic Combinatorics in Several Variables}
\label{sec:ACSV}
We gather here the main results we need from the book of Pemantle and Wilson~\cite{PeWi13}, to which we refer for more information.  Note that their results are often more general than what is presented below, allowing meromorphic functions instead of just rational ones, and asymptotics of coefficient sequences along rays other than the main diagonal. 

Given a rational function $F(\bz) = G(\bz)/H(\bz) \in \mathbb{Q}(\bz)$ with $G,H \in \mathbb{Q}[\bz]$ co-prime, whose Taylor expansion at the origin has non-empty domain of convergence~$\mathcal D$, our aim is to compute the asymptotic behavior of the integral~\eqref{eq:integral} giving the $k^\text{th}$ coefficient of its diagonal, as $k$ tends to infinity.  
\\

\paragraph{Critical Points}
A critical point is, by definition, a point satisfying Eq.~\eqref{eq:sys-critical}. A critical point~$\bzeta$ is called \emph{minimal} when it belongs to the boundary $\partial\mD$. It is \emph{smooth} when $\nabla H(\bzeta)\neq\bzer$ and \emph{non-degenerate} when the Hessian of the map~$\phi\mid_{\mathcal{V}(H)}$ is not singular at~$\bzeta$. Note that smoothness and non-degeneracy are generic properties. 

Informally, the torus $T$ in Eq.~\eqref{eq:integral} can be deformed to be arbitrarily close to a minimal critical point inside of the domain of analyticity of the integrand without changing the value of the integral. Next, if the critical point is smooth, a local analysis amounts to computing a residue, reducing to an integral over~$\mathcal{V}(H)$. If, moreover, the point is non-degenerate then a saddle-point analysis can be performed.

More precisely, when $\bzeta$ is smooth we may assume, without loss of generality, that $\partial H/\partial {z_n}(\bzeta) \neq 0$ so the implicit function theorem gives an analytic function $g(z_1,\dots,z_{n-1})$ parameterizing $z_n$ on an open neighborhood of $\bzeta$ in $\mV(H)$. Writing $\hat\bz=(z_1,\dots,z_{n-1})$, extending the contour to absorb the pole and then taking a residue at~$z_n=g(\hat\bz)$, the integrand of~\eqref{eq:integral} becomes 
\[\frac{-G(\hat\bz,g(\hat\bz))}{\frac{\partial H}{\partial z_n}(\hat\bz,g(\hat\bz))}\frac{d\hat\bz}{\phi(\hat{\bz},g(\hat\bz))^{k+1}},\]
which is then analyzed locally on $\mathcal{V}(H)$ in the neighborhood of $\hat\bz=\hat\bzeta$. Define $\psi(\hat\bz):=\phi(\hat\bz,g(\hat\bz))$. Observe that $\psi(\hat\bzeta)=\phi(\bzeta)=\bzeta^1$ and that the critical point equations are equivalent to $\nabla\psi(\hat\bzeta)=0$. A Taylor expansion then gives
\[\psi(\hat\bz):=\phi(\bzeta)+\frac12(\hat\bz-\hat\bzeta)^t\cdot\mathcal{H}(\bzeta)\cdot(\hat\bz-\hat\bzeta)+O(|\hat\bz-\hat\bzeta|^3),\]
where $\mathcal{H}$ is the Hessian matrix of~$\psi$, and the integrand is asymptotically equivalent to
\[\frac{-G(\bzeta)}{\frac{\partial H}{\partial z_n}(\bzeta)\phi(\bzeta)^{k+1}}\exp\left(-\frac{k+1}{2\phi(\bzeta)}(\hat\bz-\hat\bzeta)^t\cdot\mathcal{H}(\bzeta)\cdot(\hat\bz-\hat\bzeta)\right).\]
If $\bzeta$ is non-degenerate, this is precisely the setting for the saddle-point method, leading to the following result.

\begin{lemma}\label{lemma:formula}\cite[Th. 9.2.7, Cor. 9.2.9]{PeWi13}
Suppose $F(\bz)$ has a smooth, minimal, non-degenerate critical point at $\bzeta$ and that $E_{\bzeta}:=\mV \cap \{ \bz : |z_i|\le|\zeta_i|,1\leq i\leq n \}=\{\bzeta\}$.  Then the diagonal coefficients~$A_k$ behave like
\begin{equation}
\bzeta^{-k}k^{\frac{1-n}{2}}\left(\frac{(2\pi)^{(1-n)/2}}{\sqrt{\bzeta^{3-n}|\mathcal{H}(\bzeta)|/\zeta_n^2}}\frac{-G(\bzeta)}{\zeta_n\frac{\partial H}{\partial_{z_n}}(\bzeta)}+O\left(\frac{1}{k}\right)\right),\label{eq:formula}
\end{equation}
where $|\mathcal{H}|$ is the determinant of the Hessian of $\psi$ above. When $E_{\bzeta}$ contains a finite number of non-degenerate minimal critical points, one simply sums up the contributions given by each point. The branch of the square-root $\sqrt{|\mathcal{H}(\bzeta)|}$ is determined by taking the product of the principal branch square-roots of the eigenvalues of the Hessian $\mathcal{H}(\bzeta)$.
\end{lemma}
If $G(\bzeta)=0$ then the leading term in Eq.~\eqref{eq:formula} will vanish. One can often determine dominant behaviour by calculating more terms of the asymptotic expansion~\cite[Cor. 13.3.3]{PeWi13}.

\paragraph{Simplifying Assumptions}
Following Lemma~\ref{lemma:formula} we make the following assumptions throughout the rest of the article:

\begin{itemize} \itemsep=2mm
\item[(A1)] The map $|\phi|:\overline{\mD}\rightarrow\mathbb{R}$ achieves its supremum at a unique point in $\overline{\mD}$ (see~\cite[Def. 8.1.4]{PeWi13});
\item[(A2)] $\mV(H)$ is smooth ($H$ and its partial derivatives do not simultaneously vanish);
\item[(A3)] all critical points of $\phi$ are non-degenerate.
\end{itemize}

Note that (A3) implies that there are a finite number of critical points, as any non-degenerate critical point is isolated. In Section~\ref{sec:CPs} we make an additional algebraic assumption (A) on an ideal encoding the critical points. This assumption is slightly stronger than requiring a finite number of critical points, but is necessary for our precise complexity results.

The text of Pemantle and Wilson~\cite{PeWi13} always assumes (A1) and that $\phi$ admits a critical point in $\mathbb{C}^n$.  Although it includes results when there are no critical points on $\partial\mD$, they do not yield explicit asymptotic formulas in the same automatic manner as the minimal critical point case.  The results there rely on isolated critical points, and all explicit general asymptotic results in dimension $n>2$ need non-degenerate critical points.  Chapters 10 and 11 develop the theory when $\mV(H)$ is not smooth; extending the results of this work to cover non-smooth cases is ongoing work.

\paragraph{Combinatorial Case}
The geometric situation becomes much simpler when the coefficients $a_{\bk}$ of the generating function are nonnegative, which is known as the \emph{combinatorial} case. When this occurs, absolute convergence leads to the following analogue of Pringsheim's theorem.
\begin{lemma}\label{lemma:combin}\cite[Prop. 8.4.3]{PeWi13}
\label{prop:combCP}
If $a_{\bk} \geq 0$ for all but a finite set of values then
\begin{compactitem}
\item[--] $\phi$ admits a minimal critical point unless the supremum of $|\phi|$ is not achieved on $\overline{\mD}$;
\item[--] under (A2), if $\phi$ admits a minimal critical point then it admits one with all positive coordinates; if $\phi$ admits two minimal critical points with positive coordinates then it admits an infinite number of them.
\end{compactitem}
\end{lemma}

\section{Symbolic-Numeric ToolKit}\label{sec:symnum}
\subsection{Polynomials and their roots}
Our algorithms reduce all computations to manipulations of univariate or bivariate polynomials. We use several algorithms and classical bounds which we recall here.
\\

\paragraph{Height}
The \emph{height} $h(P)$ of a polynomial $P\in\mathbb{Z}[\bz]$ is the maximum of 0 and the logarithms in base~2 of the absolute values of the coefficients of $P$. Thus $h(P)$ gives a bound on the bit size of the coefficients. (Note that some authors define the height as the exponential of this one.)
The height satisfies the following well-known inequalities.
\begin{lemma}\label{lemma:height} 
For $P_1,\dots,P_{k},P,Q \in \mathbb{Z}[z]$,
\begin{align*}
  h(P_{1} + \cdots +P_{k} ) &\leqslant \max_{i}  h ( P_{i} ) + \log_2  k,\\
  h(P_{1} \dotsm P_{k} ) &\leqslant \sum_{i} {h ( P_{i} )} + \sum_{i=1}^{k-1} {\deg P_{i}} ,\\
  h(P)&\leqslant \deg P + h(PQ)+\log_2\sqrt{\deg(PQ)+1}.
\end{align*}
\end{lemma}
The first two are direct consequences of the definition; the last one is Mignotte's bound~\cite{Mignotte1992}.
\\

\paragraph{Bounds on roots of polynomials}
The following useful bounds are classical.
\begin{lemma}\cite{Mignotte1992} \label{lemma:roots} 
Let $A \in \mathbb{Z}[T]$ be a polynomial of degree $d\geqslant 2$ and height $h$. If $A(\alpha)=0$, then
\begin{compactitem}
  \item[--] if $\alpha\neq0$, then $1/(2^h+1)\le |\alpha|\le 2^h+1$;
  \item[--] if $A(\beta)=0$ and $\alpha\neq\beta$, then $|\alpha-\beta|\ge d^{-(d+2)/2}\|A\|^{1-d}$;
  \item[--] if $Q(\alpha)\neq0$ for $Q\in\mathbb{Z}[T]$, then
  \[ |Q(\alpha)|\ge ((d+1)2^{h(Q)})^{-d+1}2^{-h\deg Q};\]
  \item[--] if $A$ is square-free then
   $|A'(\alpha)| \geqslant d^{-d} \|A\|^{-2d+1}$,
\end{compactitem}
where $\|A\|$ is the 2-norm of the vector of coefficients, easily bounded by $2^h\sqrt{d+1}$.
\end{lemma}

\paragraph{Bounds on resultants and gcds}

\begin{lemma}\label{lemma:resultant}
For $P$ and $Q$ in ${\mathbb Z}[U]$ of height at most $h$ and degree at most $d$, $\gcd(P,Q)$ has height $\tilde{O}(d+h)$ and can be computed in $\tilde{O}(d^2+dh)$ bit operations. For $P$ and $Q$ in $\mathbb{Z}[T,U]$, let
\begin{align*} 
\delta&=\deg_UP\deg_TQ+\deg_UQ\deg_TP \\
\eta &= h(P)\deg_UQ+ h(Q)\deg_UP+\log((\deg_UP+\deg_UQ)!),
\end{align*}
then $\deg\Res_U(P,Q)\le \delta$, $h(\Res_U(P,Q))\le\eta$.
\end{lemma}
\begin{proof}
The height of the gcd is given by Lemma~\ref{lemma:height} and the univariate complexity is well known~\cite[Cor. 11.17]{GathenGerhard2003}. 
The bounds for the resultant 
follow from a direct expansion of the determinant of the Sylvester matrix for $\Res_U(P,Q)$ using Lemma~\ref{lemma:height} again. 
\end{proof}

\paragraph{Algorithms for roots of polynomials}
\begin{lemma}\cite{SagraloffMehlhorn2016},\cite{MehlhornSagraloffWang2015},\cite[Th.~10]{KobelSagraloff2015}\label{lemma:fsolve} For a square-free  $A\in{\mathbb Z}[T]$ of degree $d$ and height $h$,
\begin{itemize} \itemsep=2mm
  \item[--] isolating intervals (resp. disks) of radius less than $2^{-\kappa}$ for all real (resp. complex) roots of $A$ can be computed in $\tilde{O} (d^{3} +d^{2} h +d\kappa)$ bit operations;
  \item[--] for real (resp. complex) $t$, computing real (resp. complex) $a_t$ such that $|A(t)-a_{t}|<2^{-\ell}$ can be achieved in $\tilde{O}(d (h+\ell+d\log\max(1,|t|) ))$ bit operations given $\tilde{O}(h+\ell+d\log\max(1,|t|))$ bits of $t$;
  \item[--] simultaneously computing such approximations at points $t_{1},\dots,t_{N}$ with $N=O(d)$ can be achieved within the same bound with $|t|$ replaced by $\max |t_i|$.
\end{itemize}
\end{lemma}

\begin{cor}\label{cor:numgcd}
Given a polynomial $A\in\mathbb{Z}[T]$ as above, isolating intervals (or disks) of radius less than $2^{-O(d(h+d))}$ for its real (or complex) roots and a factor $P\in\mathbb{Z}[T]$ of $A$, then selecting which real (or complex) roots of $A$ are roots of $P$ can be achieved in $\tilde{O}(d^3+d^2h)$ bit operations.
\end{cor}
\begin{proof}By Lemma~\ref{lemma:height}, the height of $P$ is bounded by $\deg P+h+\log_2\sqrt{d+1}$. Then the previous lemma gives isolating intervals or disks for the roots of $P$ in~$\tilde{O}(d^3+d^2h+d\kappa)$ bit operations and the separation bound of Lemma~\ref{lemma:roots} shows that $\kappa=O(d(h+d))$ is sufficient to associate a unique isolating disk or interval to each root of $A$.
\end{proof}

\paragraph{Grouping Roots by Modulus}
Unlike the good separation bound given in Lemma~\ref{lemma:roots} between distinct complex roots of a polynomial, which decreases like~$2^{-hd}$, the best separation bound for the moduli of roots that we know of~\cite[Th. 1]{GourdonSalvy1996} gives an order of~$2^{-hd^3}$, which is not sufficient for our purposes. Grouping roots by identical modulus can be done with the following instead.
\begin{lemma}\label{lemma:modsep}For $A\in\mathbb{Z}[T]$ of degree $d\ge2$ and height $h$, if $A(\alpha)=0$ and $A(|\alpha|)\neq 0$, then
\[A(|\alpha|)A(-|\alpha|)\ge((d^2+1)2^{2h})^{-d^2+1}2^{-2d^2(h+\log d)}.\]
\end{lemma}
\begin{proof}
The Graeffe polynomial $G(T):=A(\sqrt{T})A(-\sqrt{T})$ has degree $d$, height at most $2h$ and its positive real roots are the square of the real roots of $A$. By Lemma~\ref{lemma:resultant}, the resultant $R(u)=\Res_T(A(T),T^dA(u/T))$ has degree at most $d^2$ and height bounded by $2hd+2d\log d$. It vanishes at the products $\alpha\beta$ of roots of $A$ and in particular at the square $|\alpha|^2=\alpha\overline\alpha$. The conclusion follows directly from the third statement of Lemma~\ref{lemma:roots}.
\end{proof}
\begin{cor}\label{cor:moduli}
With the same notation, if $0<r_1\le\dots\le r_k$ are the real positive roots of $A$ then all roots of $A$ of modulus exactly $r_1,\dots,r_k$ can be computed, with isolating disks of radius~$2^{-\tilde{O}(hd^2)}$, in $\tilde{O}(hd^3)$ bit operations.
\end{cor}
\begin{proof}
 By Lemmas~\ref{lemma:roots} and~\ref{lemma:fsolve}, we can determine all complex roots $\alpha$ of $A$ to~$\tilde{O}(hd^2)$ bits and decide which are real and positive in~$\tilde{O}(hd^3)$ bit operations.  Furthermore, Lemma~\ref{lemma:roots} implies that each $|\alpha|$ is bounded by $2^{O(h)}$ so that evaluating $A$ simultaneously at the moduli of all the roots and their negatives at precision $\tilde{O}(hd^2)$ also costs $\tilde{O}(hd^3)$ bit operations.  When $A(|\alpha|)A(-|\alpha|)$ is smaller than the bound of Lemma~\ref{lemma:modsep}, then $A(|\alpha|)=0$ and $|\alpha|$ can be identified with one of the real roots of $A$ since it is known to a larger precision than the separation bound from Lemma~\ref{lemma:roots}.
\end{proof}
In practice, one would first compute roots only at precision $\tilde{O}(hd)$, in $\tilde{O}(hd^2)$ bit operations, and then check whether any of the nonreal roots has a modulus that could equal one of the real positive roots in view of its isolating interval. Only those roots need to be refined to higher precision (say, by Newton iteration) before invoking Lemma~\ref{lemma:modsep}. 

\subsection{Numerical Kronecker Representation}
In the next section we locate the dominant critical points using a combination of rigorous numerical evaluations, via Lemma~\ref{lemma:fsolve}, and algorithms for polynomial systems, for which we require precise bounds on degrees and heights. Our basic data-structure is a Kronecker representation. We now develop the required bounds and algorithms.
\\

\paragraph{Kronecker Representation}
Let $\bff=(f_1,\dots,f_n)$ and $g$ be polynomials in $\QQ[\bz]$. Assume that $\bff$ forms a \emph{reduced regular sequence} in $\mathbb{C}\setminus\mathcal{V}(g)$. This means that for $i=1,\dots,n-1$, the saturation $\mathcal{I}_i$ with respect to $g$ of the ideal generated by $(f_1,\dots,f_{i-1})$ is radical and $f_i$ is not a zero divisor modulo $\mathcal{I}_i$. A consequence is that the system $f_1=\dots=f_n$ with $g\neq0$ has finitely many solutions. A \emph{Kronecker representation} of these solutions is a system:
\begin{equation}
  P(u)=0 ,\qquad 
    \left\{\begin{array}{ll}
    P'(u)z_1 - Q_1(u)&=0 ,\\ 
    &\hspace{0.04in}\vdots \\
    P'(u)z_n - Q_n(u)&=0,\end{array}\right. \label{eq:resolution}
\end{equation}
where $u = \sum_{i=1}^n \lambda_i z_i$ is a linear form in the $z_i$'s with integer coefficients that takes a different value at each distinct solution $\bz$ of the system, $P$ is a square-free polynomial in $\mathbb{Z}[u]$ and $Q_1,\ldots,Q_{n}$ are in $\mathbb{Z}[u]$ of degree bounded by that of $P$.

It is often convenient to consider the polynomials $f_1,\dots,f_n,$ $g$ of degree at most $d$ as given by a straight-line program (a program using only assignments, constants and arithmetic operations~\cite{BurgisserClausenShokrollahi1997}) which evaluates them simultaneously at any~$\bz$ using at most $L$ arithmetic operations. A pessimistic bound on $L$ is obtained by considering $n+1$ dense polynomials in $n$ variables, leading to $L=\tilde{O}(D)$. (Recall that $D=d^n$.)

The main properties of the Kronecker representation, summarizing results due to Giusti, Heintz, Lecerf, Pardo, Schost and collaborators, are as follows.
\begin{prop}
  \label{lemma:geomres}
Let $f_1,\dots,f_n,g$ be polynomials in $\mathbb{Z}[\bz]$ of degree at most $d$ and height at most $h$, given by a straight-line program of size at most $L$ and integers of height $\tilde{O}(h)$, such that $\bff$ forms a reduced regular sequence in $\mathbb{C}^n\setminus\mathcal{V}(g)$. Then a Kronecker representation of the solutions to $f_1=\dots=f_n=0$ where $g\neq0$ of the form {\eqref{eq:resolution}} exists, where the coefficients of the linear form $u$ are bounded by $O(D^2)$, the degree of $P$ (and each $Q_i$) is bounded by $D$ and the heights of $P$ and the $Q_{i}$ are bounded by $\tilde{O}( hD )$. It can be computed by a probabilistic algorithm in $\tilde{O} ( hLD^{2} )=\tilde{O}(hD^3)$ bit operations.
\end{prop}

The probabilistic aspects are harmless in practice; we refer to the literature for a discussion of bounds on the probability of error and probabilistic checks~\cite[\S2.1]{GiustiLecerfSalvy2001},\cite[\S15.6]{Schost2001}.

\begin{proof}By B\'ezout's bound the system has at most $D$ solutions. Thus a generic linear form separates these solutions if the product of the $\binom{D}{2}$ differences of its evaluations at its solutions is nonzero. This is a polynomial of degree~$\binom{D}{2}$ in the coefficients of the linear form and the conclusion on its height follows from the Zippel-Schwarz lemma~\cite[Lem. 6.44]{GathenGerhard2003}. The algorithm and its arithmetic complexity of~$\tilde{O}(LD^2)$ are given by Giusti \emph{et alii}~\cite{GiustiLecerfSalvy2001}.
The bound on the height of the coefficients is due to Schost~\cite[Th.~12]{Schost2001}. A modular computation can be performed using a prime of length $\tilde{O}(\log(hD^2))$~\cite[Prop. 44, Cor. 4]{Schost2001} and computing such a prime is not expensive \cite[Th. 18.10]{GathenGerhard2003}. The bit complexity of that stage follows by multiplying this logarithm by the arithmetic complexity, and the modular result is Newton lifted to~$\mathbb{Z}$ in $\tilde{O}(hLD^2)$ bit operations~\cite[Th. 2]{GiustiLecerfSalvy2001}, dominating the complexity.
\end{proof}

\paragraph{Reduction}
The analogue of Gr\"obner base reduction takes the following form, which controls the height of the result.

\begin{lemma}\label{lemma:redkro}
With the same hypotheses, given a polynomial $q\in\mathbb{Z}[\bz]$ of degree at most~$d$ and height bounded by $h$, there exists a parameterization $P'(u)-TQ_q(u)$ of the values taken by $q$ on the solutions of~\eqref{eq:resolution}, of degree smaller than $D$ and height in $\tilde{O}(hD)$. Moreover, there exists a polynomial~$\Phi_q\in\mathbb{Z}[T]$ of degree at most~$D$ and height~$\tilde{O}(hD^2)$ that vanishes at the values taken by $q$ on the solutions. Given a straight-line program of evaluation complexity $\ell=\tilde{O}(D)$, $Q_q$ can be computed in~$\tilde{O}(\ell hD^2)$ bit operations, and $\Phi_q$ in~$\tilde{O}(hD^2(D+\ell))$ bit operations. If $q$ is linear, the height of $\Phi_q$ is $\tilde{O}(hD)$ and the complexity of finding $\Phi_q$ and $Q_q$ drops to $\tilde{O}(hD^2)$.
\end{lemma}

\begin{proof}
Consider the extended system $(f_1,\dots,f_n,t-q)$ for a new variable~$t$. As $t-q$ is linear and monic in $t$, the number of solutions is unchanged and the separating linear form $u$ for the original system is still separating for the new system. The degree and height of $Q_q$ follow from Proposition~\ref{lemma:geomres}.

The computation of the polynomial $P'(u)-TQ_q(u)$ parameterizing the values of $q$ at the solutions can be obtained in~$\tilde{O}(\ell D)$ arithmetic (not bit) operations by expanding
$q(Q_1(u)/P'(u),\dots,Q_n(u)/P'(u))$
 as a power series at precision $2D$ and multiplying by $P'(u)$. The bit complexity is obtained by Hensel lifting~\cite{GiustiLecerfSalvy2001}.
 
By Lemma~\ref{lemma:resultant}, the height of the resultant of $P(u)$ and $P'(u)-TQ_q(u)$ is~$\tilde{O}(hD^2)$ and by Lemma~\ref{lemma:height} this is also a bound on the height of its factors, in particular the minimal polynomial of multiplication by $Q_q/P'$ in~$\mathbb{Q}[u]/P(u)$. 
In the case when~$q$ has degree~1, the key point is that the bound on the height of the polynomial~$\Phi_q$ is governed by bounds on the height of the variety defined by Eq.~\eqref{eq:resolution}~\cite{KrickPardoSombra2001,DahanSchost2004,Schost2001}. This variety has dimension~0, degree bounded by $D$ and height at most~$\tilde{O}(hD)$, which implies that its image by an affine $q$ also has height of that order~\cite[Prop. 2.4]{KrickPardoSombra2001}. 

The polynomial $\Phi_q$ can then be computed by modular computation modulo sufficiently many primes and Chinese remaindering using the efficient minimal polynomial algorithm of Kedlaya and Umans~\cite[\S8.4]{KedlayaUmans2011}.
\end{proof}

In most of our uses, the existence of the polynomial $\Phi$ with such bounds is important but we do not need it explicitly.
\\

\paragraph{Numerical Kronecker Representation}
We now combine the results of Lemma~\ref{lemma:fsolve} and Proposition~\ref{lemma:geomres}. Since the linear form $u$ of the Kronecker representation is a linear form in the $z_i$ with integer coefficients, it takes a real value when the coordinates of a solution do, and conversely. Thus the real roots of $P$ are in one-to-one correspondence with the real solutions of the system.

A \emph{numerical Kronecker representation} of the real (resp. complex) roots of a system is given by a Kronecker representation and isolating intervals (resp. disks) for the roots of the polynomial~$P$. Selecting the solutions of the system having a particular property then means selecting the isolating intervals (resp. disks) of the corresponding roots of~$P$.

\begin{lemma}\label{lemma:numKronecker}
With the same hypotheses as Proposition~\ref{lemma:geomres}, a numerical Kronecker representation of the real (resp. complex) solutions with isolating intervals (resp. disks) of radius $2^{-\kappa}$ can be computed in~$\tilde{O}(D(hD^2+\kappa))$ bit operations.
\end{lemma}
\begin{proof}Combine the bounds on $P$ from Proposition~\ref{lemma:geomres} and the results summarized in Lemma~\ref{lemma:fsolve}.
\end{proof}

\paragraph{Zero coordinates and sign}
The numerical Kronecker representation allows for \emph{exact} symbolic-numeric computation. We begin with zero testing of coordinates: as $P$ is square-free, its derivative is not~0 at its roots and thus detecting whether a coordinate of a point parameterized by a root $\alpha$ of $P$ is~0 reduces to detecting whether the polynomial~$Q_i$ vanishes at~$\alpha$.  To do this one can compute the $\gcd$ of $P$ and $Q_i$ and then identify their common roots, working numerically with sufficient precision. 

\begin{lemma}\label{lemma:zero}
With the same hypotheses as Proposition~\ref{lemma:geomres} detecting the coordinates of solutions that are exactly~0 can be performed in~$\tilde{O}(hD^3)$ bit operations, and so can computing the sign of the coordinates of real solutions.
\end{lemma}
\begin{proof}
 By Lemma~\ref{lemma:resultant}, the $\gcd$ of $P$ and $Q_i$ has height $\tilde{O}(hD)$ and can be computed in $\tilde{O}(hD^2)$ bit operations. Corollary~\ref{cor:numgcd} and Lemma~\ref{lemma:numKronecker} then show that $\tilde{O}(hD^3)$ bit operations are sufficient to compute the roots of $P$ cancelling~$Q_i$ at a sufficient precision $\kappa=\tilde{O}(hD^2)$. 
Evaluating $P'$ and the $Q_i$ at the same precision then gives the sign of the non-zero coordinates of real solutions by Lemma~\ref{lemma:roots}.
\end{proof}

\paragraph{Numerical evaluation of the solutions}
In summary, we have obtained the following result.
\begin{theorem}\label{th:fsolvekro} 
With the same hypotheses as Proposition~\ref{lemma:geomres}, computing isolating intervals (resp. disks) of radius $2^{-\kappa}$ for all coordinates of all real (resp. complex) solutions of the system, recognizing the zero coordinates and computing the signs of the coordinates of the real solutions can be performed in a total of~$\tilde{O}(D(hD^2+\kappa))$ bit operations. 
\end{theorem}
\begin{proof}
The precision on the roots of $P$ needed to decide that a coordinate is not zero and find the sign of the coordinates of real solutions is $\tilde{O}(hD^2)$ bits, as shown in the proof of Lemma~\ref{lemma:zero}. For any $\kappa\ge0$, computing the roots of $P$ to precision $2^{-\kappa-\tilde{O}(hD^2)}$ is sufficient to evaluate the coordinates up to precision $2^{-\kappa}$. By Lemma~\ref{lemma:fsolve}, this numerical resolution of $P$ is achieved in~$\tilde{O}(D(hD^2+\kappa))$ bit operations and the simultaneous evaluation of $P'$ and the $Q_i$ to this precision again obeys that complexity bound. 
\end{proof}

\paragraph{Equal coordinates}
Deciding equality can also be done exactly, again within the same complexity bound. This amounts to detecting an index~$i$ and two distinct roots of $P$ giving the same values to $Q_i/P'$.

\begin{lemma}\label{lemma:equal}
With the same hypotheses as Proposition~\ref{lemma:geomres}, detecting which solutions have an identical coordinate can be performed in $\tilde{O}(hD^3)$ bit operations.
\end{lemma}
\begin{proof}
By Lemma~\ref{lemma:redkro}, there exists a polynomial $\Phi_i(T)\in\mathbb{Z}[T]$ that cancels the values of $Q_i(u)/P'(u)$ mod $P(u)$ and has degree at most $D$ and height $\tilde{O}(hD)$. Thus, by Lemma~\ref{lemma:roots}, the distance between distinct values of $Q_i/P'$ at solutions cannot be smaller than~$2^{-\tilde{O}(hD^2)}$. 
By Lemma~\ref{lemma:roots}, $\log|1/P'(\upsilon)|$ $=\tilde{O}(hD^2)$ at the roots $\upsilon$ of $P$, so knowing $Q_i(\upsilon)$ and $P'(\upsilon)$ to $\tilde{O}(hD^2)$ bits is sufficient to obtain $Q_i(\upsilon)/P'(\upsilon)$ to $\tilde{O}(hD^2)$ bits. 
By Lemma~\ref{lemma:fsolve}, computing $Q_i$ and $P'$ at all real (resp. complex) roots of $P$ with that precision can be achieved in~$\tilde{O}(hD^3)$ bit operations given the roots of $P$ at precision $\kappa=\tilde{O}(hD^2)$, which is also computable in that complexity. 
Adding the costs of this computation for each $Q_i$ does not change the $\tilde{O}$ complexity estimate. 
\end{proof}

\section{Contributing Critical Points}
\label{sec:CPs}
We now apply this machinery to the computation of minimal critical points in analytic combinatorics.

\subsection{Computation of the Critical Points}
The critical points of the map $\phi:(\bz)\mapsto z_1\dotsm z_n$ are computed from System~\eqref{eq:sys-critical}. 
We make the following assumption:\\

\begin{itemize}
  \item[(A)] \emph{System \eqref{eq:sys-critical} is a reduced regular sequence in
   $\mathbb{C}^{n+1}\setminus\mathcal{V}(\lambda)$.}\\
\end{itemize}

As $\lambda=z_jH_{z_j}$ for each $j=1,\dots,n$ the condition $z_j \neq 0$ follows from $\lambda \neq0$. Furthermore, a non-smooth point of the variety is marked by the simultaneous vanishing of all the partial derivatives, so that by removing~$\mathcal{V}(\lambda)$ we are only considering the smooth critical points.

\begin{lemma} \label{lemma:critKronecker}
Under assumption (A), a Kronecker representation of the solutions of System~\eqref{eq:sys-critical} completed by the equation $T-\phi(\bz)$ is of the form~\eqref{eq:resolution}, with two more parameterizations $P'(u)-\lambda Q_\lambda(u)$, and $P'(u)-T Q_\phi(u)$; the polynomial $P$ and the $Q_i$, 
$Q_\lambda$ and $Q_\phi$ have degree at most $D$ and height~$\tilde{O}(hD)$. The Kronecker representation can be computed in~$\tilde{O}(hD^3)$ bit operations; an accompanying numerical Kronecker representation in $\tilde{O}(hD^3)$ bit operations.
\end{lemma}

\begin{proof}
Under assumption (A), the system formed by $H$ and the $n-1$ polynomials $z_iH_{z_i}-z_nH_{z_n}$ for $i=1,\dots,n-1$ satisfies the hypotheses of Proposition~\ref{lemma:geomres}, giving the desired degree, height and complexity bounds for the $z_i$ variables. The polynomials $Q_\lambda$ and $Q_\phi$ are obtained by Lemma~\ref{lemma:redkro} and the numerical Kronecker representation from Theorem~\ref{th:fsolvekro}.
\end{proof}

\subsection{The Minimal Positive Critical Point}
In the combinatorial case, by absolute convergence, minimality of a point $\bzeta$ with positive coordinates means that the line segment from the origin to $\bzeta$ does not intersect $\mathcal{V}(H)$.
Furthermore, by Lemma~\ref{lemma:combin} if there are a finite number of critical points then there is a unique minimal critical point whose coordinates are positive real numbers.

Thus, the question of minimality can be restated as finding the critical point with positive coordinates $\bzeta$ such that the polynomial $H(t\zeta_1,\dots,t\zeta_n)$ has no roots in the interval $(0,1)$. This in turn is equivalent to the absence of a root of the polynomial $A(u,t) =P'(u)^d H(tz(u))$ in the interval $(0,1)$ when $u$ is evaluated at the root of $P$ corresponding to $\bzeta$. We could not find a way to exploit Sturm sequences without increasing the exponent of $D$ in the complexity to at least~4. Instead, we compute a Kronecker representation for System~\eqref{eq:sys-critical} enriched with the equation $H(tz_1,\dots,tz_n)$ for a new variable~$t$ and gather the critical points as above.

\begin{lemma}\label{lemma:minimal}
Suppose assumptions~(A1)--(A3) and (A) are satisfied, and that $G/H$ is combinatorial.  Let $\bzeta$ be the unique \emph{minimal} critical point with positive real coordinates.  Then one can find isolating intervals of width less than~$2^{-\kappa}$ for each coordinate of $\bzeta$ using a probabilistic algorithm running in~$\tilde{O}(h(dD)^3+D\kappa)$ bit operations.
\end{lemma}

\begin{proof}
Since System~\eqref{eq:sys-critical} has finitely many solutions, so does the one with $H(tz_1,\dots,tz_n)$ added in. The linear form used in~\eqref{eq:resolution} may not separate the solutions anymore as multiple values of $t$ correspond to one previous solution, so we use a new linear form involving $t$ as well, which we denote~$v$ to avoid confusion. The results of Proposition~\ref{lemma:geomres} show the existence of a Kronecker representation with a squarefree univariate polynomial~$\tilde{P}(v)=0$ and parameterizations of the form $\tilde{P}'(v)z_i-\tilde{Q}_i(v)=0$ for $i$ in $\{1,\dots,n,\lambda,T,t\}$, with $\tilde{P},\tilde{Q}_i$ of degree at most $dD$ and all heights bounded by $\tilde{O}(hdD)$. It can be computed in~$\tilde{O}(h(dD)^3)$ bit operations.

By construction, the solutions of the original system are recovered when $t=1$, so $\tilde{P}'(v)-\tilde{Q}_t(v)$ has a gcd with $\tilde{P}$ corresponding to the (at most~$D$) critical points. This gcd can be computed in~$\tilde{O}(hd^2D^2)$ bit operations~\cite[Cor. 11.11]{GathenGerhard2003}. By Corollary~\ref{cor:numgcd}, we can recover the roots of $P$ corresponding to solutions of the original system in~$\tilde{O}(h(dD)^3)$ bit operations. Theorem~\ref{th:fsolvekro} allows us to compute the signs of $t$ and $1-t$ on the other roots in~$\tilde{O}(h(dD)^3)$ bit operations.
 \end{proof}

\subsection{Critical Points on the Same Torus}
Next, Lemma~\ref{lemma:formula} shows that the other critical points lying on the same torus (i.e., having the same coordinate-wise modulus) may contribute to the dominant term of the asymptotic behavior. In the combinatorial case, we can again make use of the numerical Kronecker representation to obtain them in good complexity, although this is a more expensive operations than the previous ones.

\begin{prop} \label{prop:contrib}
Suppose assumptions~(A1)--(A3) and (A) are satisfied, and that $G/H$ is combinatorial.  Then isolating disks of radius less than $2^{-\kappa}$ can be computed for each coordinate of \emph{all} minimal critical points by a probabilistic algorithm running in~$\tilde{O}(D(hD^3+\kappa))$ bit operations.
\end{prop}

\begin{proof}
Let $\bzeta$ be the unique minimal critical point with positive real coordinates. By Lemma~\ref{lemma:redkro}, the minimal polynomials $M_1,\dots,M_n$ of the $z_i$ coordinates have degree at most $D$, height $\tilde{O}(hD)$ and can all be computed in~$\tilde{O}(hD^2)$ bit operations. 
As $\zeta_i$ is a real positive root of $M_i$ for $i=1,\dots,n$, by Lemma~\ref{lemma:modsep}, if $|\sigma_i|=\zeta_i$ then~$M_i(\sigma_i)=0$, which can be detected by evaluating $M(|\sigma_i|)M(-|\sigma_i|)$ with $\tilde{O}(hD^3)$ bits of precision. By Theorem~\ref{th:fsolvekro} with $\kappa=\tilde{O}(hD^3)$, this can be achieved in~$\tilde{O}(hD^4)$ bit operations, which is also the cost for evaluating the $M_i$ simultaneously at the roots of~$P$.
\end{proof}
This is the most costly operation we use, but in many cases it is sufficient to check whether $\phi(\bz)$ takes different values at each critical point. As the polynomial~$T-\phi(\bz)$ is part of the original system, this can be checked in~$\tilde{O}(hD^3)$ bit operations by Lemma~\ref{lemma:equal}. When two critical points give the same value to~$\phi$, then the test above is used.

\subsection{Degenerate Critical Points}
Simple but tedious multivariate calculus shows that the $(i,j)$ entry of the Hessian matrix $\mathcal{H}$ of Lemma~\ref{lemma:formula} is
\[
\frac{\phi(\bzeta)}{\lambda\zeta_i\zeta_j}(U_{i,n}+U_{j,n}-U_{i,j}-U_{n,n}-\lambda),\quad
 U_{k,\ell}:=
\zeta_k\zeta_\ell\frac{\partial^2H}{\partial z_k\partial z_\ell}(\bzeta),
\]
when $i\neq j$ and the formula for a diagonal element is the same with the term $-\lambda$ replaced by $-2\lambda$. 

Thus by multilinearity, the determinant is of the form $\phi(\bzeta)^{n-3}\zeta_n^2/\lambda^{n-1}$ times the determinant of a matrix of polynomials of degree at most~$d$ and height $O(h+\log d)$. The prefactor is not~0 in view of our hypotheses and will cancel out other factors in the final formula from Lemma~\ref{lemma:formula}. By Lemma~\ref{lemma:redkro} each entry of the remaining determinant reduces to a polynomial in~$u$ of degree smaller than $D$ and height~$\tilde{O}(hD)$, divided by $P'(u)$ that can be taken out of the determinant to yield a leading factor of~$P'(u)^{1-n}$. The remaining determinant is a polynomial of degree smaller than $nD$ and height~$\tilde{O}(hD)$ that can be computed by evaluation-interpolation in~$\tilde{O}(hD^2)$ bit operations. By Lemma~\ref{lemma:fsolve}, it can be evaluated at the roots of $P$ with error smaller than $2^{-\kappa}$ in~$\tilde{O}(D(hD^2+\kappa))$ bit operations given~$\tilde{O}(hD^2+\kappa)$ bits of the roots of~$P$. In summary, we have proved the following.
\begin{lemma}\label{lemma:nondegenerate} Under assumption~(A), the degeneracy of the critical points can be tested in~$\tilde{O}(hD^3)$ bit operations; the determinant of the Hessian of Lemma~\ref{lemma:formula} can be computed with error less than~$2^{-\kappa}$ in~$\tilde{O}(D(hD^2+\kappa))$ bit operations.
\end{lemma}

\subsection{Dominant Asymptotics}
Once the minimal critical points have been obtained, the computation is completed by evaluating the formula from Lemma~\ref{lemma:formula}. The following theorem is a direct consequence of the previous lemmas. To lighten notation we assume that $G(\bz)$ has height $O(h)$ and degree $O(d)$.

\begin{theorem} \label{th:main}
Suppose assumptions (A1)--(A3) and (A) hold, and $F(\bz)$ is combinatorial.  If $F(\bz)$ admits a single minimal critical point $\bzeta$ then 
\[ A_k = \left( T^{-k} \cdot k^{(1-n)/2} \cdot (2\pi)^{(1-n)/2} \right) \left(C+O(1/k)\right),\] 
where $T$ and $C$ are algebraic constants. Isolating disks for $T$ and $C$ of radius less than $2^{-\kappa}$ can be obtained by a probabilistic algorithm in~$\tilde{O}(h(dD)^3+D\kappa)$ bit operations.

If $F(\bz)$ admits more than one minimal critical point then this can be detected in~$\tilde{O}(hD^4)$ bit operations; each has a value of $T$ with the same modulus, and $A_k$ is asymptotically equal to the sum of the contributions of each point, obtained within the same complexity as above.
\end{theorem}

Note that all the steps of the algorithm have a complexity bound of~$\tilde{O}(h(dD)^3)$ bit operations except for the test that critical points lie on the same torus, for which easy filters are available in practice (see Example~\ref{ex:pattern-avoiding-2}). If $G(\bzeta)=0$, then one can usually find dominant asymptotics by taking further terms of the relevant Taylor expansions at $\bzeta$.

\section{Experiments}
\label{sec:Experiments}

We list here several examples highlighting the above techniques and detail some of the steps taken by our implementation.
Note that while the theoretical bounds were obtained using an algorithm that computes the Kronecker representation directly~\cite{GiustiLecerfSalvy2001}, our preliminary implementation relies on a fast implementation of Gr{\"o}bner bases~\cite{Faugere2010}\footnote{The code and Maple worksheets for these examples are available at \url{http://diagasympt.gforge.inria.fr}.}. Although our theoretical bounds do not include explicit constants, in practice one has the actual polynomials that arise in the computation and can therefore calculate the accuracy needed to perform the operations of the numerical Kronecker method.

\begin{example}
\label{ex:apery2}
We first detail the computations in the case of Ap\'ery's sequence from Example~\ref{ex:apery1}. The polynomial
\[H(a,b,c,z) = 1-z(1+a)(1+b)(1+c)(abc+bc+b+c+1)\]
of degree~7 defines a smooth variety, and each term in the power series expansion of $1/H$ is non-negative. A Kronecker representation of the solutions of the system
\[aH_a-\lambda,bH_b-\lambda,cH_c-\lambda,zH_z-\lambda,T-abcz,H(a,b,c,z)\]
with the linear form $u=T$ (before using a random linear form, we try a few simple ones that, if they do separate the solutions, lead to smaller integers) is $u^2-34u+1=0,$
\begin{equation}\label{eq:crit-Apery}
a=\frac{2u-82}{2u-34},b=c=\frac{24}{2u-34},z=\frac{-164u+4}{2u-34},T=u.
\end{equation}
Thus, there are two (real) critical points and a simple numerical evaluation shows that the only one with positive coordinates corresponds to $u=17-12\sqrt{2}$. 

Next, we prove automatically that there is no intersection of $\mathcal{V}(H)$ with the segment from 0 to that critical point. The linear form $T$ does not separate the solutions of the system with $H(ta,tb,tc,tz)$ replacing $H$, but the linear form $u=T+t$ does. It leads to a Kronecker representation with polynomial $\tilde{P}(u)$ of degree~14, with a factor of degree~2 that corresponds to the critical points and is recovered by taking the gcd of $P$ with $P'-Q_t$. The other factor has four real roots, but none of them gives a value of $t$ in $(0,1)$.

Since both critical points are real, it is easily checked numerically that they do not lie on the same torus.
Then, the evaluation of the asymptotic behaviour using Eq.~\eqref{eq:crit-Apery} gives
\[\frac{u^{-k}}{\pi^{3/2}k^{3/2}}\frac{\sqrt{3462-102u}}{48}(1+O(1/k)).\]
Finally, injecting $u=17-12\sqrt{2}$ gives the result from Eq.~\eqref{eq:asymptApery}.
\end{example}

\begin{example}\cite[\S5.3]{PeWi08}\label{ex:pattern-avoiding-2}
We give more details on Example~\ref{ex:pattern-avoiding}.
By its combinatorial origin, the rational function $F$ is combinatorial and one can check that its denominator defines a smooth variety. The linear form $u=T$ is separating. The Kronecker representation gives a polynomial $P$ of degree~21 and a parameterization $P'(u) x - Q_1(u)=P'(u) y - Q_2(u)=0$, with $Q_1$ and $Q_2$ of degree 20 in $u$ with largest coefficient about $2500$. As with most examples here it exhibits the nice behaviour of the Kronecker representation with respect to the sizes of the integers~\cite{Rouillier1999}: if we used the representation $x = \tilde{Q}_1(u)$ with polynomial $\tilde{Q}_1$, then $\tilde{Q}_1$ has rational coefficients whose largest numerator is about $10^{20}$.

There are 3 real positive roots of $P$, which are candidates for the minimal critical point with positive coordinates (that only the positive roots are relevant is due to our choice of linear form). Of these, two give points $(x,y)$ with positive coordinates: $u_1 \approx 0.255$ and $u_2 \approx 2.792$.
We can prove that $u_1$ is minimal by adding the equation $H(tx,ty)=0$ to our system and getting the new Kronecker representation; using a Gr{\"o}bner Basis computation takes too long, but one can use a subresultant calculation to compute the new Kronecker representation in under a minute on modern laptops. 

Next, in this example, in order to check that the critical point given by $u_1$ is the only critical point on its torus it is sufficient to evaluate numerically the parameterization of $T$ at the roots of $P$. This gives a cheap filter that eliminates many cases and avoids the more costly computation of Lemma~\ref{lemma:modsep}. A further simplification here is that since $T=u$, it is sufficient to compare the moduli of the complex roots of $P$ with $u_1$ and it is found that indeed, there are no other critical point on the torus. Lemma~\ref{lemma:formula} then gives the asymptotics displayed in Example~\ref{ex:pattern-avoiding}.
\end{example}

\begin{example}
The rational function
\[F(x,y) = \frac{1}{(1-x-y)(20-x-40y)-1},\]
has a smooth denominator and is combinatorial (factor $(1-x-y)$ out of the denominator). The Kronecker computation finds two critical points with positive coordinates: $(x_1,y_1)\approx (0.548, 0.309)$ and $(x_2,y_2)\approx (9.997, 0.252).$ Since $x_1<x_2$ and $y_1>y_2$, just by examining the critical points we cannot determine which is the critical minimal point. Adding the polynomial~$H(tx,ty)$ to the system shows that there is a point with approximate coordinates $(0.092x_2,0.092y_2)$ in $\mV$, so that $x_1$ is the minimal critical point.  To three decimal places the diagonal asymptotics have the form 
\[A_k = (5.88\ldots)^k k^{-1/2} (0.054\ldots)\cdot(1+O(1/k)).  \]
\end{example}

\section{Conclusion}
\label{sec:Conclusion}

This is only the beginning. The numerical Kronecker representation seems to be an interesting (and useful) tool. The next step will be to extend these results by removing some of the regularity constraints on which it relies. 

In terms of analytic combinatorics our results give the first known complexity bounds. Much work remains to be done to extend to asymptotics off the main diagonal, to degenerate or non-smooth cases, to relax the regularity conditions and to see whether the non-combinatorial case can also be brought into the same (singly exponential) complexity class.

\paragraph{Acknowledgments}
The authors thank \'E.~Schost for help with some results in his thesis.
This work has been supported in part by FastRelax ANR-14-CE25-0018-01, NSERC, the French Ministry of Foreign Affairs, and the France Canada Research Fund.

\bibliographystyle{plain}
\bibliography{bibl}

\begin{thebibliography}{10}

\bibitem{BPR06}
S.~Basu, R.~Pollack, and M.-F. Roy.
\newblock {\em Algorithms in real algebraic geometry}.
\newblock Springer-Verlag, Berlin, 2006.

\bibitem{BostanLairezSalvy2013}
Alin Bostan, Pierre Lairez, and Bruno Salvy.
\newblock Creative telescoping for rational functions using the
  {G}riffiths-{D}work method.
\newblock In {\em ISSAC '13}, pages 93--100. ACM Press, 2013.

\bibitem{BostanLairezSalvy2015}
Alin Bostan, Pierre Lairez, and Bruno Salvy.
\newblock Multiple binomial sums.
\newblock {\em J. Symbolic Comput.}, (To appear), 2016.

\bibitem{BouzidiLazardPougetRouillier2015}
Yacine Bouzidi, Sylvain Lazard, Marc Pouget, and Fabrice Rouillier.
\newblock Separating linear forms and rational univariate representations of
  bivariate systems.
\newblock {\em J. Symbolic Comput.}, 68(part 1):84--119, 2015.

\bibitem{BurgisserClausenShokrollahi1997}
Peter B{\"u}rgisser, Michael Clausen, and M.~Amin Shokrollahi.
\newblock {\em Algebraic complexity theory}, volume 315 of {\em Grundlehren der
  Mathematischen Wissenschaften}.
\newblock Springer-Verlag, 1997.

\bibitem{DahanSchost2004}
Xavier Dahan and \'{E}ric Schost.
\newblock Sharp estimates for triangular sets.
\newblock In {\em ISSAC '04}, pages 103--110. ACM Press, 2004.

\bibitem{DevHoPe11}
T.~DeVries, J.~van~der Hoeven, and R.~Pemantle.
\newblock Automatic asymptotics for coefficients of smooth, bivariate rational
  functions.
\newblock {\em Online J. Anal. Comb.}, 6:24 pages, 2011.

\bibitem{EmirisMourrainTsigaridas2010}
Ioannis~Z. Emiris, Bernard Mourrain, and Elias~P. Tsigaridas.
\newblock The {DMM} bound: multivariate (aggregate) separation bounds.
\newblock In {\em I{SSAC} 2010}, pages 243--250. ACM, New York, 2010.

\bibitem{Faugere2010}
Jean-Charles Faug{\`e}re.
\newblock {FGb: A Library for Computing Gr{\"o}bner Bases}.
\newblock In {\em {ICMS 2010}}, volume 6327 of {\em LNCS}, pages 84--87.
  Springer, 2010.

\bibitem{FlSe09}
P.~Flajolet and R.~Sedgewick.
\newblock {\em Analytic Combinatorics}.
\newblock Cambridge University Press, 2009.

\bibitem{GiustiHeintzMoraisMorgensternPardo1998}
M.~Giusti, J.~Heintz, J.~E. Morais, J.~Morgenstern, and L.~M. Pardo.
\newblock Straight-line programs in geometric elimination theory.
\newblock {\em J. Pure Appl. Algebra}, 124(1-3):101--146, 1998.

\bibitem{GiustiLecerfSalvy2001}
Marc Giusti, Gr{\'e}goire Lecerf, and Bruno Salvy.
\newblock A {G}r{\"o}bner free alternative for polynomial system solving.
\newblock {\em J.Complexity}, 17(1):154--211, 2001.

\bibitem{GourdonSalvy1996}
Xavier Gourdon and Bruno Salvy.
\newblock Effective asymptotics of linear recurrences with rational
  coefficients.
\newblock {\em Discr. Math.}, 153(1--3):145--163, 1996.

\bibitem{KedlayaUmans2011}
Kiran~S. Kedlaya and Christopher Umans.
\newblock Fast polynomial factorization and modular composition.
\newblock {\em SIAM J. Comput.}, 40(6):1767--1802, 2011.

\bibitem{KobelSagraloff2015}
Alexander Kobel and Michael Sagraloff.
\newblock On the complexity of computing with planar algebraic curves.
\newblock {\em J. Complexity}, 31(2):206--236, 2015.

\bibitem{KrickPardoSombra2001}
Teresa Krick, Luis~Miguel Pardo, and Mart{\'{\i}}n Sombra.
\newblock Sharp estimates for the arithmetic {N}ullstellensatz.
\newblock {\em Duke Math. J.}, 109(3):521--598, 2001.

\bibitem{MehlhornSagraloffWang2015}
Kurt Mehlhorn, Michael Sagraloff, and Pengming Wang.
\newblock From approximate factorization to root isolation with application to
  cylindrical algebraic decomposition.
\newblock {\em J. Symbolic Comput.}, 66:34--69, 2015.

\bibitem{Mignotte1992}
Maurice Mignotte.
\newblock {\em Mathematics for Computer Algebra}.
\newblock Springer New York, 1992.

\bibitem{Mora2003}
Teo Mora.
\newblock {\em Solving polynomial equation systems. {I}}, volume~88 of {\em
  Encyclopedia of Mathematics and its Applications}.
\newblock Cambridge University Press, Cambridge, 2003.
\newblock The Kronecker-Duval philosophy.

\bibitem{PeWi08}
R.~Pemantle and M.~C. Wilson.
\newblock Twenty combinatorial examples of asymptotics derived from
  multivariate generating functions.
\newblock {\em SIAM Review}, 50:199--272, 2008.

\bibitem{PeWi13}
R.~Pemantle and M.~C. Wilson.
\newblock {\em Analytic Combinatorics in Several Variables}.
\newblock Cambridge University Press, 2013.

\bibitem{Raic12}
A.~Raichev.
\newblock amgf documentation -- release 0.8.
\newblock \url{https://github.com/araichev/amgf}, 2012.

\bibitem{Rouillier1999}
Fabrice Rouillier.
\newblock Solving zero-dimensional systems through the rational univariate
  representation.
\newblock {\em Appl. Algebra Engrg. Comm. Comput.}, 9(5):433--461, 1999.

\bibitem{SagraloffMehlhorn2016}
Michael Sagraloff and Kurt Mehlhorn.
\newblock Computing real roots of real polynomials.
\newblock {\em J. Symbolic Comput.}, 73:46--86, 2016.

\bibitem{Schost2001}
{{\'E}}ric Schost.
\newblock {\em Sur la r{\'e}solution des syst{\`e}mes polynomiaux {\`a}
  param{\`e}tres}.
\newblock PhD thesis, {\'E}cole polytechnique, 2001.

\bibitem{GathenGerhard2003}
Joachim von~zur Gathen and J{\"u}rgen Gerhard.
\newblock {\em Modern computer algebra}.
\newblock Cambridge University Press, 2003.

\end{thebibliography}

\end{document}